\documentclass{fac}

\usepackage{graphicx}
\usepackage{amssymb}
\usepackage{amsfonts}
\usepackage{amsmath}
\usepackage{dsfont}
\usepackage{MnSymbol}
\usepackage{mathtools}
\usepackage{epsfig}
\usepackage{epstopdf}
\usepackage{tikz}
\usepackage{amsthm}
\ifprodtf \else \usepackage{latexsym}\fi

\newcommand\black{\ensuremath{\blacktriangleright}}
\newcommand\white{\ensuremath{\vartriangleright}}

\newif\ifamsfontsloaded
\ifprodtf
  \newcommand\whbl{\white\kern-.1em--\kern-.1em\black}
  \newcommand\blwh{\black\kern-.1em--\kern-.1em\white}
  \newcommand\blbl{\black\kern-.1em--\kern-.1em\black}
  \newcommand\whwh{\white\kern-.1em--\kern-.1em\white}
  \amsfontsloadedtrue
\else
  \checkfont{msam10}
  \iffontfound
    \IfFileExists{amssymb.sty}
      {\usepackage{amssymb}\amsfontsloadedtrue
       \newcommand\whbl{\white\kern-.125em--\kern-.125em\black}%
       \newcommand\blwh{\black\kern-.125em--\kern-.125em\white}%
       \newcommand\blbl{\black\kern-.125em--\kern-.125em\black}%
       \newcommand\whwh{\white\kern-.125em--\kern-.125em\white}}
      {}
  \fi
\fi

\newtheorem{theorem}{Theorem}[section]
\newtheorem{definition}[theorem]{Definition}

\title[Draft of Operational Semantics of Games]
      {Operational Semantics of Games}

\author[Yong Wang]
    {Yong Wang\\
     College of Computer Science and Technology,\\
     Faculty of Information Technology,\\
     Beijing University of Technology, Beijing, China\\
     }

\correspond{Yong Wang, Pingleyuan 100, Chaoyang District, Beijing, China.
            e-mail: wangy@bjut.edu.cn}

\pubyear{2019}
\pagerange{\pageref{firstpage}--\pageref{lastpage}}

\begin{document}
\label{firstpage}

\makecorrespond

\maketitle

\begin{abstract}
We introduce operational semantics into games. And based on the operational semantics, we establish a full algebra of games, including basic algebra of games, algebra of concurrent games, recursion and abstraction. The algebra can be used widely to reason on the behaviors of systems (not only computational systems) with game theory supported.
\end{abstract}


\section{Introduction}{\label{int}}

Game theory has been widely used to interpret the nature of the world. The combination of game theory and (computational) logic \cite{LIG} always exists two ways.

One is to use game theory to interpret computational logic, such as the well-known game semantics \cite{PCF} \cite{PCF2} \cite{PCF3} \cite{MIL} \cite{Algol}, in which game theory acts as a foundational semantics bases to understand the behaviors of computer programming language.

The other is to give game theory a logic basis, such as game logic \cite{GL1} \cite{GL2} \cite{GL3}, game algebras \cite{BAG1} \cite{BAG2}, algebras \cite{CG4} for concurrent games
\cite{CG1} \cite{CG2} \cite{CG3}.

In this paper, we introduce operational semantics into games, and based on the operational semantics, we establish a fully algebraic axiomatization of games, including the basic algebra of games, algebra of concurrent games, recursion and abstraction. This paper is organized as follows. In Section \ref{osg}, we introduce operational semantics into games. We introduce the basic algebra of games, algebra of concurrent games, recursion and abstraction in Section \ref{bag}, \ref{acg}, \ref{rec} and \ref{abs}, respectively. Finally, in Section \ref{con}, we conclude this paper.

\section{Operational Semantics of Games}\label{osg}

In this section, we introduce the related equational logic, and structured operational semantics of games, which serve as the bases of game algebras. The concrete equational logics and operational semantics of games are included in the follow algebras of games.

\subsection{Proof Techniques}\label{PT}

\begin{definition}[Game language]\label{gl}
The game language $GL$ consists of:
\begin{enumerate}
  \item a set of atomic games $\mathcal{G}_{at}=\{g_a\}_{a\in A}$, and a special \emph{idle} atomic game $\iota=g_0\in\mathcal{G}_{at}$;
  \item game operations, including choice of first player $\vee$, choice of second player $\wedge$, dualization $^d$, composition of games $\circ$, and parallel of games $\parallel$.
\end{enumerate}

Atomic games and their duals are called literals. And models of $GL$ are called game boards.
\end{definition}

\begin{definition}[Game terms]\label{gt}
The game terms are defined inductively as follows:
\begin{itemize}
  \item every atomic game $g_a$ is a game term;
  \item if $G,H$ are game terms, then $G^d, H^d$, $G\vee H$, $G\wedge H$, $G\circ H$ and $G\parallel H$ are all game terms.
\end{itemize}
\end{definition}

\begin{definition}[Elimination property]
Let a game algebra with a defined set of basic terms as a subset of the set of closed terms over the game algebra. Then the game algebra has the elimination to basic terms property if for every closed term $G$ of the algebra, there exists a basic term $H$ of the algebra such that the algebra$\vdash G=H$.
\end{definition}

\begin{definition}[Strongly normalizing]
A term $G_0$ is called strongly normalizing if does not an infinite series of reductions beginning in $G_0$.
\end{definition}

\begin{definition}
We write $G>_{lpo} H$ if $G\rightarrow^+ H$ where $\rightarrow^+$ is the transitive closure of the reduction relation defined by the transition rules of a game algebra.
\end{definition}

\begin{theorem}[Strong normalization]\label{SN}
Let a term rewriting (TRS) system with finitely many rewriting rules and let $>$ be a well-founded ordering on the language of the corresponding algebra. If $G>_{lpo} H$ for each rewriting rule $G\rightarrow H$ in the TRS, then the term rewriting system is strongly normalizing.
\end{theorem}

\subsection{Labeled Transition System}\label{lts}

\begin{definition}[Labeled transition system]
A transition is a triple $(s,g_a,s')$ with $g_a\in \mathcal{G}_{at}$, or a pair (s, P) with $P$ a predicate, where $s,s'\in S$ of states. A labeled transition system (LTS) is possibly infinite set of transitions. An LTS is finitely branching if each of its states has only finitely many outgoing transitions.
\end{definition}

\begin{definition}[Transition system specification]
A transition rule $\varrho$ is an expression of the form $\frac{\varpi}{\pi}$, with $\varpi$ a set of expressions $G\xrightarrow{g_a}^i G'$ with $a\in A;i=1,2$, and $tP$ with $G,G'$ are game terms, called the (positive) premises of $\varrho$, and $\pi$ an expression $H\xrightarrow{g_a}^i H'$ or $GP$ with $H,H'$ are game terms, called the conclusion of $\varrho$. The left-hand side of $\pi$ is called the source of $\varrho$. A transition rule is closed if it does not contain any variables. A transition system specification (TSS) is a (possible infinite) set of transition rules.
\end{definition}

\begin{definition}[Congruence]
An equivalence relation $\mathcal{E}$ on game board $B$ is a congruence if for each $f\in GL$, if $G_i\mathcal{E} H_i$ for $i\in\{1,\cdots,ar(f)\}$, then $f(G_1,\cdots,G_{ar(f)})\mathcal{E} f(H_1,\cdots,H_{ar(f)})$.
\end{definition}

\begin{definition}[Conservative extension]
Let $T_0$ and $T_1$ be TSSs over $GL_0$ and $GL_1$, respectively. The TSS $T_0\oplus T_1$ is a conservative extension of $T_0$ if the LTSs generated by $T_0$ and $T_0\oplus T_1$ contain exactly the same transitions $G\xrightarrow{a}^i G'$ and $GP$ with the game term $G$.
\end{definition}

\begin{definition}[Source-dependency]
The source-dependent variables in a transition rule of $\varrho$ are defined inductively as follows: (1) all variables in the source of $\varrho$ are source-dependent; (2) if $G\xrightarrow{g_a}^d G'$ is a premise of $\varrho$ and all variables in $G$ are source-dependent, then all variables in $G'$ are source-dependent. A transition rule is source-dependent if all its variables are. A TSS is source-dependent if all its rules are.
\end{definition}

\begin{definition}[Freshness]
Let $T_0$ and $T_1$ be TSSs over signatures $GL_0$ and $GL_1$, respectively. A term in $\mathbb{T}(T_0\oplus T_1)$ is said to be fresh if it contains a function symbol from $GL_1\setminus GL_0$. Similarly, a transition label or predicate symbol in $T_1$ is fresh if it does not occur in $T_0$.
\end{definition}

\begin{theorem}[Conservative extension]
Let $T_0$ and $T_1$ be TSSs over $GL_0$ and $GL_1$, respectively, where $T_0$ and $T_0\oplus T_1$ are positive after reduction. Under the following conditions, $T_0\oplus T_1$ is a conservative extension of $T_0$. (1) $T_0$ is source-dependent. (2) For each $\varrho\in T_1$, either the source of $\varrho$ is fresh, or $\varrho$ has a premise of the form $G\xrightarrow{g_a}^d G'$ or $GP$, where $G$ is a game term, all variables in $G$ occur in the source of $\varrho$ and $G'$, $g_a$ or $P$ is fresh.
\end{theorem}

\subsection{Game Equivalence}\label{ge}

\begin{definition}[Outcome conditions]
$GL=\langle S, \{\rho^i_a\}_{a\in A;i=1,2}\rangle$ are called game boards, where $S$ is the set of states and $\rho^i_a\subseteq S\times P(S)$ are outcome relations, which satisfy the following two forcing conditions:

\begin{enumerate}
  \item monotonicity (MON): for any $s\in S$, and $X\subseteq Y\subseteq S$, if $s\rho^i_a X$, then $s\rho^i_a Y$;
  \item consistency (CON): for any $s\in S, X\subseteq S$, if $s\rho^1_a X$, then not $s\rho^2_a (S-X)$.
\end{enumerate}

And the following optional conditions:

\begin{enumerate}
  \item termination (FIN): for any $s\in S$, then $s\rho^i_a S$, and the class of terminating game boards are denoted $\textbf{FIN}$;
  \item determinacy (DET): $s\rho^2_a (S-X)$ iff $s\rho^1_a X$, and the class of determined game boards are denoted $\textbf{DET}$.
\end{enumerate}

The outcome relation $\rho^i_G$ for any game term $G$ can be defined inductively according to the structure of $G$.
\end{definition}

\begin{definition}[Game equivalence]\label{dge}
For game terms $G_1$ and $G_2$ on game board $B$, if $\rho^i_{G_1}\subseteq \rho^i_{G_2}$, then $G_1$ is i-included in $G_2$ on $B$, denoted $G_1\subseteq_i g_2$; if $G_1\subseteq_1 G_2$ and $G_1\subseteq_2 G_2$, then $G_1$ is included in $G_2$ on $B$, denoted $B\models G_1\preceq G_2$; if $B\models G_1\preceq G_2$ for any $B$, then $G_1\preceq G_2$ is called a valid term inclusion, denoted $\models G_1\preceq G_2$.

If $G_1$ and $G_2$ are assigned the same outcome relation in $B$, then they are game equivalent on $B$, denoted $B\models G_1\sim G_2$; if $B\models G_1\sim G_2$ for any game board $B$, then $G_1\sim G_2$ is a valid term identity, denoted $\models G_1 \sim G_2$.

It is easy to see that game equivalence is an equivalent relation.
\end{definition}

\begin{definition}[Weak game equivalence]\label{dwge}
For game terms $G_1$ and $G_2$ on game board $B$, $\iota_{I_G}$ with $I\subseteq G$ renames all $g_a\in I_G$ into $\iota$, if $\rho^i_{\iota_{I_{G_1}}(G_1)}\subseteq \rho^i_{\iota_{I_{G_2}}(G_2)}$, then $G_1$ is weakly i-included in $G_2$ on $B$, denoted $G_1\sqsubseteq_i g_2$; if $G_1\sqsubseteq_1 G_2$ and $G_1\sqsubseteq_2 G_2$, then $G_1$ is weakly included in $G_2$ on $B$, denoted $B\models G_1\ll G_2$; if $B\models G_1\ll G_2$ for any $B$, then $G_1\ll G_2$ is called a valid weak term inclusion, denoted $\models G_1\ll G_2$.

If $\iota_{I_{G_1}}(G_1)$ and $\iota_{I_{G_2}}(G_2)$ are assigned the same outcome relation in $B$, then they are weak game equivalent on $B$, denoted $B\models G_1\approx G_2$; if $B\models G_1\approx G_2$ for any game board $B$, then $G_1\approx G_2$ is a valid term identity, denoted $\models G_1 \approx G_2$.

It is easy to see that weak game equivalence is an equivalent relation.
\end{definition}

\section{Basic Algebra of Games}\label{bag}

In this section, we will discuss Basic Algebra of Games, abbreviated BAG, which include game operations: choice of the first player $\vee$ ($\vee^1$), choice of the second player $\wedge$ ($\vee^2$), dualization $^d$ and composition of games $\circ$.

\subsection{Axiom System of BAG}

In the following, let $g_a, g_b, g_a', g_b'\in \mathcal{G}_{at}$, and let variables $x,y,z$ range over the set of game terms, $G,H$ range over the set of closed terms. The set of axioms of BAG consists of the laws given in Table \ref{AxiomsForBAG}.

\begin{center}
    \begin{table}
        \begin{tabular}{@{}ll@{}}
            \hline No. &Axiom\\
            $G1$ & $x\vee x = x\quad x\wedge x =x$\\
            $G2$ & $x\vee y = y\vee x\quad x\wedge y = y\wedge x$\\
            $G3$ & $x\vee (y\vee z)=(x\vee y)\vee z\quad x\wedge (y\wedge z)=(x\wedge y)\wedge z$\\
            $G4$ & $x\vee(x\wedge y)=x\quad x\wedge (x\vee y)=x$\\
            $G5$ & $x\vee(y\wedge z)=(x\vee y)\wedge (x\vee z)\quad x\wedge (y\vee z)=(x\wedge y)\vee (x\wedge z)$\\
            $G6$ & $(x^d)^d = x$\\
            $G7$ & $(x\vee y)^d= x^d\wedge y^d\quad (x\wedge y)^d = x^d\vee y^d$\\
            $G8$ & $(x\circ y)\circ z = x\circ(y\circ z)$\\
            $G9$ & $(x\vee y)\circ z = (x\circ z)\vee (y\circ z)\quad (x\wedge y)\circ z = (x\circ z)\wedge (y\circ z)$\\
            $G10$ & $x^d\circ y^d=(x\circ y)^d$\\
%            $G11$ & $y\preceq z \rightarrow x\circ y\preceq x\circ z$\\
            $G11$ & $x\circ \iota = \iota\circ x = x$\\
            $G12$ & $\iota^d = \iota$\\
        \end{tabular}
        \caption{Axioms of BAG}
        \label{AxiomsForBAG}
    \end{table}
\end{center}

\subsection{Properties of BAG}

\begin{definition}[Basic terms of BAG]\label{BTBAG}
The set of basic terms of BAG, $\mathcal{B}(BAG)$, is inductively defined as follows:
\begin{enumerate}
  \item $\mathcal{G}_{at}\subset\mathcal{B}(BAG)$;
  \item $\mathcal{G}_{at}^d\subset\mathcal{B}(BAG)$;
  \item if $g_a\in \mathcal{G}_{at}, G\in\mathcal{B}(BAG)$ then $g_a\circ G\in\mathcal{B}(BAG)$;
  \item if $G,H\in\mathcal{B}(BAG)$ then $G\vee H\in\mathcal{B}(BAG)$;
  \item if $G,H\in\mathcal{B}(BAG)$ then $G\wedge H\in\mathcal{B}(BAG)$.
\end{enumerate}
\end{definition}

\begin{theorem}[Elimination theorem of BAG]\label{ETBAG}
Let $G$ be a closed BAG term. Then there is a basic BAG term $H$ such that $BAG\vdash G=H$.
\end{theorem}

\begin{proof}
(1) Firstly, suppose that the following ordering on the signature of BAG is defined: $^d > \circ > \wedge >\vee$ and the symbol $^d$ is given the lexicographical status for the first argument, then for each rewrite rule $G\rightarrow H$ in Table \ref{TRSForBAG} relation $G>_{lpo} H$ can easily be proved. We obtain that the term rewrite system shown in Table \ref{TRSForBAG} is strongly normalizing, for it has finitely many rewriting rules, and $>$ is a well-founded ordering on the signature of BAG, and if $G>_{lpo} H$, for each rewriting rule $G\rightarrow H$ is in Table \ref{TRSForBAG}.

\begin{center}
    \begin{table}
        \begin{tabular}{@{}ll@{}}
            \hline
            No. & Rewriting Rule\\
            $RG1$ & $x\vee x \rightarrow x\quad x\wedge x \rightarrow x$\\
            $RG3$ & $x\vee (y\vee z)\rightarrow (x\vee y)\vee z\quad x\wedge (y\wedge z)\rightarrow (x\wedge y)\wedge z$\\
            $RG4$ & $x\vee(x\wedge y)\rightarrow x\quad x\wedge (x\vee y)\rightarrow x$\\
            $RG5$ & $x\vee(y\wedge z)\rightarrow (x\vee y)\wedge (x\vee z)\quad x\wedge (y\vee z)\rightarrow(x\wedge y)\vee (x\wedge z)$\\
            $RG6$ & $(x^d)^d \rightarrow x$\\
            $RG7$ & $(x\vee y)^d\rightarrow x^d\wedge y^d\quad (x\wedge y)^d \rightarrow x^d\vee y^d$\\
            $RG8$ & $(x\circ y)\circ z \rightarrow x\circ(y\circ z)$\\
            $RG9$ & $(x\vee y)\circ z \rightarrow (x\circ z)\vee (y\circ z)\quad (x\wedge y)\circ z \rightarrow (x\circ z)\wedge (y\circ z)$\\
            $RG10$ & $x^d\circ y^d\rightarrow (x\circ y)^d$\\
            $RG11$ & $x\circ \iota \rightarrow x \quad \iota\circ x \rightarrow x$\\
            $RG12$ & $\iota^d \rightarrow \iota$\\
        \end{tabular}
        \caption{Term rewrite system of BAG}
        \label{TRSForBAG}
    \end{table}
\end{center}

(2) Then we prove that the normal forms of closed BAG terms are basic BAG terms.

Suppose that $G$ is a normal form of some closed BAG term and suppose that $G$ is not a basic term. Let $G'$ denote the smallest sub-term of $G$ which is not a basic term. It implies that each sub-term of $G'$ is a basic term. Then we prove that $G$ is not a term in normal form. It is sufficient to induct on the structure of $G'$:

\begin{itemize}
  \item Case $G'\equiv g_a, g_a\in \mathcal{G}_{at}$. $G'$ is a basic term, which contradicts the assumption that $G'$ is not a basic term, so this case should not occur.
  \item Case $G'\equiv g_a^d, g_a^d\in \mathcal{G}_{at}$. $G'$ is a basic term, which contradicts the assumption that $G'$ is not a basic term, so this case should not occur.
  \item Case $G'\equiv G_1\circ G_2$. By induction on the structure of the basic term $G_1$:
      \begin{itemize}
        \item Subcase $G_1\in \mathcal{G}_{at}$. $G'$ would be a basic term, which contradicts the assumption that $G'$ is not a basic term;
        \item Subcase $G_1\equiv g_a\circ G_1'$. $RG8$ rewriting rule can be applied. So $G$ is not a normal form;
        \item Subcase $G_1\equiv G_1'\wedge G_1''$. $RG9$ rewriting rule can be applied. So $G$ is not a normal form;
        \item Subcase $G_1\equiv G_1'\vee G_1''$. $RG9$ rewriting rule can be applied. So $G$ is not a normal form.
      \end{itemize}
  \item Case $G'\equiv G_1\wedge G_2$. By induction on the structure of the basic terms both $G_1$ and $G_2$, all subcases will lead to that $G'$ would be a basic term, which contradicts the assumption that $G'$ is not a basic term.
  \item Case $G'\equiv G_1\vee G_2$. By induction on the structure of the basic terms both $G_1$ and $G_2$, all subcases will lead to that $G'$ would be a basic term, which contradicts the assumption that $G'$ is not a basic term.
\end{itemize}
\end{proof}

\subsection{Structured Operational Semantics of BAG}

In this subsection, we will define a term-deduction system which gives the operational semantics of BAG. We give the operational transition rules for atomic games, $\iota$, game operations $^d$, $\vee^d$ and $\circ$ as Table \ref{TRForBAG} shows. And the predicate $\xrightarrow{g_a}\surd$ represents successful termination after playing of the game $g_a$, the predicate $\xrightarrow{g_a}^d\surd$ represents successful termination after playing of the game $g_a$ by the player $d$.

\begin{center}
    \begin{table}
        $$\frac{}{g_a^d\xrightarrow{g_a}^d\surd}\quad \frac{}{\iota\rightarrow\surd}$$
        $$\frac{x\xrightarrow{g_a}^d\surd}{x\vee^d y\xrightarrow{g_a}^d\surd} \quad\frac{x\xrightarrow{g_a}^dx'}{x\vee^d y\xrightarrow{g_a}^d x'} \quad\frac{y\xrightarrow{g_a}^d\surd}{x\vee^d y\xrightarrow{g_a}^d\surd} \quad\frac{y\xrightarrow{g_a}^d y'}{x\vee^d y\xrightarrow{g_a}^d y'}$$
        $$\frac{x\xrightarrow{g_a}\surd}{x\circ y\xrightarrow{g_a} y} \quad\frac{x\xrightarrow{g_a}x'}{x\circ y\xrightarrow{g_a}x'\circ y}$$
        \caption{Transition rules of BAG}
        \label{TRForBAG}
    \end{table}
\end{center}

\begin{theorem}[Congruence of BAG with respect to game equivalence]
Game equivalence $\sim$ is a congruence with respect to BAG.
\end{theorem}

\begin{proof}
It is sufficient to prove that game equivalence is preserved by the game operations: $^d$, $\vee$, $\wedge$ and $\circ$.

(1) Case of $^d$. Suppose that $G_1\sim G_2$, it suffices to prove $G_1^d\sim G_2^d$. It can be immediately gotten from the definition of game equivalence (see Definition \ref{dge}).

(2) Case of $\vee$. Suppose that $G_1\sim G_2$ and $H_1\sim H_2$, it suffices to prove that $G_1\vee H_1\sim G_2\vee H_2$. It can be immediately gotten from the definition of game equivalence (Definition \ref{dge}) and transition rules of $\vee$ in Table \ref{TRForBAG}.

(3) Case of $\wedge$. Suppose that $G_1\sim G_2$ and $H_1\sim H_2$, it suffices to prove that $G_1\wedge H_1\sim G_2\wedge H_2$. It can be immediately gotten from the definition of game equivalence (Definition \ref{dge}) and transition rules of $\wedge$ in Table \ref{TRForBAG}.

(4) Case of $\circ$. Suppose that $G_1\sim G_2$ and $H_1\sim H_2$, it suffices to prove that $G_1\circ H_1\sim G_2\circ H_2$. It can be immediately gotten from the definition of game equivalence (Definition \ref{dge}) and transition rules of $\circ$ in Table \ref{TRForBAG}.
\end{proof}

\begin{theorem}[Soundness of BAG modulo game equivalence]
Let $x$ and $y$ be BAG terms. If $BAG\vdash x=y$, then $x\sim y$.
\end{theorem}

\begin{proof}
Since game equivalence is both an equivalent and a congruent relation, we only need to check if each axiom in Table \ref{AxiomsForBAG} is sound modulo game equivalence, according to the definition of game equivalence (Definition \ref{dge}) and transition rules in Table \ref{TRForBAG}. The checks are left to the readers as an exercise.
\end{proof}

\begin{theorem}[Completeness of BAG modulo game equivalence]
Let $G$ and $H$ be closed BAG terms, if $G\sim H$ then $G=H$.
\end{theorem}

\begin{proof}
Firstly, by the elimination theorem of BAG, we know that for each closed BAG term $G$, there exists a closed basic BAG term $G'$, such that $BAG\vdash G=G'$, so, we only need to consider closed basic BAG terms.

The basic terms (see Definition \ref{BTBAG}) modulo associativity and commutativity (AC) of $\vee^i$ (defined by axiom $RG2$ in Table \ref{AxiomsForBAG}), and this equivalence is denoted by $=_{AC}$. Then, each equivalence class $G$ modulo AC of $\vee^i$ has the following normal form

$$G_1\vee^i\cdots\vee^i G_k$$

with each $G_i$ either an atomic game or of the form $H_1\circ H_2$, and each $G_i$ is called the summand of $G$.

Now, we prove that for normal forms $N$ and $N'$, if $N\sim N'$ then $N=_{AC}N'$. It is sufficient to induct on the sizes of $N$ and $N'$.

\begin{itemize}
  \item Consider a summand $g_a$ of $N$. Then $N\xrightarrow{g_a}\surd$, so $N\sim N'$ implies $N'\xrightarrow{g_a}\surd$, meaning that $N'$ also contains the summand $g_a$.
  \item Consider a summand $H_1\circ H_2$ of $N$. Then $N\xrightarrow{H_1}H_2$, so $N\sim N'$ implies $N'\xrightarrow{H_1}H_2'$ with $H_2\sim H_2'$, meaning that $N'$ contains a summand $H_1\circ H_2'$. Since $H_2$ and $H_2'$ are normal forms and have sizes smaller than $N$ and $N'$, by the induction hypotheses $H_2\sim H_2'$ implies $H_2=_{AC} H_2'$.
\end{itemize}

So, we get $N=_{AC} N'$.

Finally, let $G$ and $H$ be basic terms, and $G\sim H$, there are normal forms $N$ and $N'$, such that $G=N$ and $H=N'$. The soundness theorem of BAG modulo game equivalence yields $G\sim N$ and $H\sim N'$, so $N\sim G\sim H\sim N'$. Since if $N\sim N'$ then $N=_{AC}N'$, $G=N=_{AC}N'=H$, as desired.
\end{proof}

\section{Algebra of Concurrent Games}\label{acg}

In this section, we added parallelism to BAG to support concurrent games \cite{CG1} \cite{CG2} \cite{CG3} \cite{CG4}, the result algebra is called Algebra of Concurrent Games, abbreviated ACG. ACG also includes an equational logic and structured operational semantics.

\subsection{Axiom System of ACG}

In the following, let $g_a, g_b, g_a', g_b'\in \mathcal{G}_{at}$, and let variables $x,y,z$ range over the set of game terms, $G,H$ range over the set of closed terms. The set of axioms of ACG consists of the laws given in Table \ref{AxiomsForACG}.

\begin{center}
    \begin{table}
        \begin{tabular}{@{}ll@{}}
            \hline No. &Axiom\\
            $CG1$ & $(x\parallel y)\parallel z = x\parallel(y\parallel z)$\\
            $CG2$ & $g_a\parallel (g_b\circ y) = (g_a\parallel g_b)\circ y$\\
            $CG3$ & $(g_a\circ x)\parallel g_b = (g_a\parallel g_b)\circ x$\\
            $CG4$ & $(g_a\circ x)\parallel (g_b\circ y) = (g_a\parallel g_b)\circ (x\parallel y)$\\
            $CG5$ & $(x\vee y)\parallel z = (x\parallel z)\vee (y\parallel z)$\\
            $CG6$ & $x\parallel (y\vee z) = (x\parallel y)\vee (x\parallel z)$\\
            $CG7$ & $(x\wedge y)\parallel z = (x\parallel z)\wedge (y\parallel z)$\\
            $CG8$ & $x\parallel (y\wedge z) = (x\parallel y)\wedge (x\parallel z)$\\
            $CG9$ & $(x\parallel y)^d= x^d\parallel y^d$\\
            $CG10$ & $\iota\parallel x=x$\\
            $CG11$ & $x\parallel\iota = x$\\
        \end{tabular}
        \caption{Axioms of ACG}
        \label{AxiomsForACG}
    \end{table}
\end{center}

\subsection{Properties of ACG}

\begin{definition}[Basic terms of ACG]\label{BTACG}
The set of basic terms of ACG, $\mathcal{B}(ACG)$, is inductively defined as follows:
\begin{enumerate}
  \item $\mathcal{G}_{at}\subset\mathcal{B}(ACG)$;
  \item $\mathcal{G}_{at}^d\subset\mathcal{B}(ACG)$;
  \item if $g_a\in \mathcal{G}_{at}, G\in\mathcal{B}(ACG)$ then $g_a\circ G\in\mathcal{B}(ACG)$;
  \item if $G,H\in\mathcal{B}(ACG)$ then $G\vee H\in\mathcal{B}(ACG)$;
  \item if $G,H\in\mathcal{B}(ACG)$ then $G\wedge H\in\mathcal{B}(ACG)$;
  \item if $G,H\in\mathcal{B}(ACG)$ then $G\parallel H\in\mathcal{B}(ACG)$.
\end{enumerate}
\end{definition}

\begin{theorem}[Elimination theorem of ACG]\label{ETACG}
Let $G$ be a closed ACG term. Then there is a basic ACG term $H$ such that $ACG\vdash G=H$.
\end{theorem}

\begin{proof}
(1) Firstly, suppose that the following ordering on the signature of ACG is defined: $^d > \parallel>\circ > \wedge >\vee$ and the symbol $^d$ is given the lexicographical status for the first argument, then for each rewrite rule $G\rightarrow H$ in Table \ref{TRSForACG} relation $G>_{lpo} H$ can easily be proved. We obtain that the term rewrite system shown in Table \ref{TRSForACG} is strongly normalizing, for it has finitely many rewriting rules, and $>$ is a well-founded ordering on the signature of ACG, and if $G>_{lpo} H$, for each rewriting rule $G\rightarrow H$ is in Table \ref{TRSForACG}.

\begin{center}
    \begin{table}
        \begin{tabular}{@{}ll@{}}
            \hline
            No. & Rewriting Rule\\
            $RCG1$ & $(x\parallel y)\parallel z \rightarrow x\parallel(y\parallel z)$\\
            $RCG2$ & $g_a\parallel (g_b\circ y) \rightarrow (g_a\parallel g_b)\circ y$\\
            $RCG3$ & $(g_a\circ x)\parallel g_b \rightarrow (g_a\parallel g_b)\circ x$\\
            $RCG4$ & $(g_a\circ x)\parallel (g_b\circ y) \rightarrow (g_a\parallel g_b)\circ (x\parallel y)$\\
            $RCG5$ & $(x\vee y)\parallel z \rightarrow (x\parallel z)\vee (y\parallel z)$\\
            $RCG6$ & $x\parallel (y\vee z) \rightarrow (x\parallel y)\vee (x\parallel z)$\\
            $RCG7$ & $(x\wedge y)\parallel z \rightarrow (x\parallel z)\wedge (y\parallel z)$\\
            $RCG8$ & $x\parallel (y\wedge z) \rightarrow (x\parallel y)\wedge (x\parallel z)$\\
            $RCG9$ & $(x\parallel y)^d\rightarrow x^d\parallel y^d$\\
            $RCG10$ & $\iota\parallel x\rightarrow x$\\
            $RCG11$ & $x\parallel\iota \rightarrow x$\\
        \end{tabular}
        \caption{Term rewrite system of ACG}
        \label{TRSForACG}
    \end{table}
\end{center}

(2) Then we prove that the normal forms of closed ACG terms are basic ACG terms.

Suppose that $G$ is a normal form of some closed ACG term and suppose that $G$ is not a basic term. Let $G'$ denote the smallest sub-term of $G$ which is not a basic term. It implies that each sub-term of $G'$ is a basic term. Then we prove that $G$ is not a term in normal form. It is sufficient to induct on the structure of $G'$:

\begin{itemize}
  \item Case $G'\equiv g_a, g_a\in \mathcal{G}_{at}$. $G'$ is a basic term, which contradicts the assumption that $G'$ is not a basic term, so this case should not occur.
  \item Case $G'\equiv g_a^d, g_a^d\in \mathcal{G}_{at}$. $G'$ is a basic term, which contradicts the assumption that $G'$ is not a basic term, so this case should not occur.
  \item Case $G'\equiv G_1\circ G_2$. By induction on the structure of the basic term $G_1$:
      \begin{itemize}
        \item Subcase $G_1\in \mathcal{G}_{at}$. $G'$ would be a basic term, which contradicts the assumption that $G'$ is not a basic term;
        \item Subcase $G_1\equiv g_a\circ G_1'$. $RG8$ rewriting rule can be applied. So $G$ is not a normal form;
        \item Subcase $G_1\equiv G_1'\wedge G_1''$. $RG9$ rewriting rule can be applied. So $G$ is not a normal form;
        \item Subcase $G_1\equiv G_1'\vee G_1''$. $RG9$ rewriting rule can be applied. So $G$ is not a normal form.
      \end{itemize}
  \item Case $G'\equiv G_1\wedge G_2$. By induction on the structure of the basic terms both $G_1$ and $G_2$, all subcases will lead to that $G'$ would be a basic term, which contradicts the assumption that $G'$ is not a basic term.
  \item Case $G'\equiv G_1\vee G_2$. By induction on the structure of the basic terms both $G_1$ and $G_2$, all subcases will lead to that $G'$ would be a basic term, which contradicts the assumption that $G'$ is not a basic term.
  \item Case $G'\equiv G_1\parallel G_2$. By induction on the structure of the basic terms both $G_1$ and $G_2$, all subcases will lead to that $G'$ would be a basic term, which contradicts the assumption that $G'$ is not a basic term.
\end{itemize}
\end{proof}

\subsection{Structured Operational Semantics of ACG}

In this subsection, we will define a term-deduction system which gives the operational semantics of ACG. We give the operational transition rules for game operation $\parallel$ as Table \ref{TRForACG} shows.

\begin{center}
    \begin{table}
        $$\frac{x\xrightarrow{g_a}\surd\quad y\xrightarrow{g_b}\surd}{x\parallel y\xrightarrow{\{g_a,g_b\}}\surd} \quad\frac{x\xrightarrow{g_a}x'\quad y\xrightarrow{g_b}\surd}{x\parallel y\xrightarrow{\{g_a,g_b\}} x'} \quad\frac{x\rightarrow{g_a}\surd\quad y\xrightarrow{g_b}y'}{x\parallel y\xrightarrow{\{g_a,g_b\}}y'} \quad\frac{x\xrightarrow{g_a}x'\quad y\xrightarrow{g_b} y'}{x\parallel y\xrightarrow{\{g_a,g_b\}} x'\parallel y'}$$
        \caption{Transition rules of ACG}
        \label{TRForACG}
    \end{table}
\end{center}

\begin{theorem}[Generalization of ACG with respect to BAG]
ACG is a generalization of BAG.
\end{theorem}

\begin{proof}
It follows from the following three facts.

\begin{enumerate}
  \item The transition rules of BAG in section \ref{bag} are all source-dependent;
  \item The sources of the transition rules ACG contain an occurrence of $\parallel$;
  \item The transition rules of ACG are all source-dependent.
\end{enumerate}

So, ACG is a generalization of BAG, that is, BAG is an embedding of ACG, as desired.
\end{proof}

\begin{theorem}[Congruence of ACG with respect to game equivalence]
Game equivalence $\sim$ is a congruence with respect to ACG.
\end{theorem}

\begin{proof}
It is sufficient to prove that game equivalence is preserved by the game operation $\parallel$.

Suppose that $G_1\sim G_2$ and $H_1\sim H_2$, it suffices to prove that $G_1\parallel H_1\sim G_2\parallel H_2$. It can be immediately gotten from the definition of game equivalence (Definition \ref{dge}) and transition rules of $\parallel$ in Table \ref{TRForACG}.
\end{proof}

\begin{theorem}[Soundness of ACG modulo game equivalence]
Let $x$ and $y$ be ACG terms. If $ACG\vdash x=y$, then $x\sim y$.
\end{theorem}

\begin{proof}
Since game equivalence is both an equivalent and a congruent relation, we only need to check if each axiom in Table \ref{AxiomsForACG} is sound modulo game equivalence, according to the definition of game equivalence (Definition \ref{dge}) and transition rules in Table \ref{TRForACG}. The checks are left to the readers as an exercise.
\end{proof}

\begin{theorem}[Completeness of ACG modulo game equivalence]
Let $G$ and $H$ be closed ACG terms, if $G\sim H$ then $G=H$.
\end{theorem}

\begin{proof}
Firstly, by the elimination theorem of ACG, we know that for each closed ACG term $G$, there exists a closed basic ACG term $G'$, such that $ACG\vdash G=G'$, so, we only need to consider closed basic ACG terms.

The basic terms (see Definition \ref{BTACG}) modulo associativity and commutativity (AC) of $\vee^i$ (defined by axiom $RG2$ in Table \ref{AxiomsForBAG}), and this equivalence is denoted by $=_{AC}$. Then, each equivalence class $G$ modulo AC of $\vee^i$ has the following normal form

$$G_1\vee^i\cdots\vee^i G_k$$

with each $G_i$ either an atomic game or of the form

$$H_1\circ \cdots \circ H_m$$

with each $H_j$ either an atomic game or of the form

$$U_1\parallel\cdots\parallel U_n$$

with each $U_l$ an atomic game, and each $G_i$ is called the summand of $G$.

Now, we prove that for normal forms $N$ and $N'$, if $N\sim N'$ then $N=_{AC}N'$. It is sufficient to induct on the sizes of $N$ and $N'$.

\begin{itemize}
  \item Consider a summand $g_a$ of $N$. Then $N\xrightarrow{g_a}\surd$, so $N\sim N'$ implies $N'\xrightarrow{g_a}\surd$, meaning that $N'$ also contains the summand $g_a$.
  \item Consider a summand $H_1\circ H_2$ of $N$.
  \begin{itemize}
    \item if $H_1\equiv g_a'$, then $N\xrightarrow{g_a'}H_2$, so $N\sim N'$ implies $N'\xrightarrow{g_a'}H_2'$ with $H_2\sim H_2'$, meaning that $N'$ contains a summand $g_a'\circ H_2'$. Since $H_2$ and $H_2'$ are normal forms and have sizes smaller than $N$ and $N'$, by the induction hypotheses if $H_2\sim H_2'$ then $H_2=_{AC} H_2'$;
    \item if $H_1\equiv g_{a_1}\parallel\cdots\parallel g_{a_n}$, then $N\xrightarrow{\{g_{a_1},\cdots,g_{a_n}\}}H_2$, so $N\sim N'$ implies $N'\xrightarrow{\{g_{a_1},\cdots,g_{a_n}\}}H_2'$ with $H_2\sim H_2'$, meaning that $N'$ contains a summand $(g_{a_1}\parallel\cdots\parallel g_{a_n})\circ H_2'$. Since $H_2$ and $H_2'$ are normal forms and have sizes smaller than $N$ and $N'$, by the induction hypotheses if $H_2\sim H_2'$ then $H_2=_{AC} H_2'$.
  \end{itemize}
\end{itemize}

So, we get $N=_{AC} N'$.

Finally, let $G$ and $H$ be basic terms, and $G\sim H$, there are normal forms $N$ and $N'$, such that $G=N$ and $H=N'$. The soundness theorem of ACG modulo game equivalence yields $G\sim N$ and $H\sim N'$, so $N\sim G\sim H\sim N'$. Since if $N\sim N'$ then $N=_{AC}N'$, $G=N=_{AC}N'=H$, as desired.
\end{proof}

\section{Recursion}\label{rec}

In this section, we introduce recursion to capture infinite games based on ACG. We do not consider the idle game $\iota$ in this section, the full consideration of $\iota$ is placed into the next section (Section \ref{abs}).

In the following, $E,F,G$ are recursion specifications, $X,Y,Z$ are recursive variables.

\subsection{Guarded Recursive Specifications}

\begin{definition}[Recursive specification]
A recursive specification is a finite set of recursive equations

$$X_1=G_1(X_1,\cdots,X_n)$$
$$\cdots$$
$$X_n=G_n(X_1,\cdots,X_n)$$

where the left-hand sides of $X_i$ are called recursion variables, and the right-hand sides $G_i(X_1,\cdots,X_n)$ are game terms in ACG with possible occurrences of the recursion variables $X_1,\cdots,X_n$.
\end{definition}

\begin{definition}[Solution]
Games $g_1,\cdots,g_n$ are a solution for a recursive specification $\{X_i=G_i(X_1,\cdots,X_n)|i\in\{1,\cdots,n\}\}$ (with respect to game equivalence $\sim$ if $g_i\sim G_i(g_1,\cdots,g_n)$ for $i\in\{1,\cdots,n\}$.
\end{definition}

\begin{definition}[Guarded recursive specification]
A recursive specification

$$X_1=G_1(X_1,\cdots,X_n)$$
$$...$$
$$X_n=G_n(X_1,\cdots,X_n)$$

is guarded if the right-hand sides of its recursive equations can be adapted to the form by applications of the axioms in ACG and replacing recursion variables by the right-hand sides of their recursive equations,

$$(g_{11}\parallel\cdots\parallel g_{1i_1})\circ G_1(X_1,\cdots,X_n)\vee^i\cdots\vee^i(g_{k1}\parallel\cdots\parallel g_{ki_k})\circ G_k(X_1,\cdots,X_n)\vee^i(h_{11}\parallel\cdots\parallel h_{1j_1})\vee^i\cdots\vee^i(h_{1j_1}\parallel\cdots\parallel h_{lj_l})$$

where $g_{11},\cdots,g_{1i_1},g_{k1},\cdots,g_{ki_k},h_{11},\cdots,h_{1j_1},h_{1j_1},\cdots,h_{lj_l}\in \mathcal{G}_{at}$.
\end{definition}

\begin{definition}[Linear recursive specification]\label{LRS}
A recursive specification is linear if its recursive equations are of the form

$$(a_{11}\parallel\cdots\parallel a_{1i_1})X_1\vee^i\cdots\vee^i(a_{k1}\parallel\cdots\parallel a_{ki_k})X_k\vee^i(b_{11}\parallel\cdots\parallel b_{1j_1})\vee^i\cdots\vee^i(b_{1j_1}\parallel\cdots\parallel b_{lj_l})$$

where $a_{11},\cdots,a_{1i_1},a_{k1},\cdots,a_{ki_k},b_{11},\cdots,b_{1j_1},b_{1j_1},\cdots,b_{lj_l}\in \mathcal{G}_{at}$.
\end{definition}

For a guarded recursive specifications $E$ with the form

$$X_1=G_1(X_1,\cdots,X_n)$$
$$\cdots$$
$$X_n=G_n(X_1,\cdots,X_n)$$

the behavior of the solution $\langle X_i|E\rangle$ for the recursion variable $X_i$ in $E$, where $i\in\{1,\cdots,n\}$, is exactly the behavior of their right-hand sides $G_i(X_1,\cdots,X_n)$, which is captured by the two transition rules in Table \ref{TRForGR}.

\begin{center}
    \begin{table}
        $$\frac{G_i(\langle X_1|E\rangle,\cdots,\langle X_n|E\rangle)\xrightarrow{\{g_1,\cdots,g_n\}}\surd}{\langle X_i|E\rangle\xrightarrow{\{g_1,\cdots,g_n\}}\surd}$$
        $$\frac{G_i(\langle X_1|E\rangle,\cdots,\langle X_n|E\rangle)\xrightarrow{\{g_1,\cdots,g_n\}} y}{\langle X_i|E\rangle\xrightarrow{\{g_1,\cdots,g_n\}} y}$$
        \caption{Transition rules of guarded recursion}
        \label{TRForGR}
    \end{table}
\end{center}

\begin{theorem}[Conservitivity of ACG with guarded recursion]
ACG with guarded recursion is a conservative extension of ACG.
\end{theorem}

\begin{proof}
Since the transition rules of ACG are source-dependent, and the transition rules for guarded recursion in Table \ref{TRForGR} contain only a fresh constant in their source, so the transition rules of ACG with guarded recursion are a conservative extension of those of ACG.
\end{proof}

\begin{theorem}[Congruence theorem of ACG with guarded recursion]
Game equivalence $\sim$ is a congruence with respect to ACG with guarded recursion.
\end{theorem}

\begin{proof}
It follows the following two facts:
\begin{enumerate}
  \item in a guarded recursive specification, right-hand sides of its recursive equations can be adapted to the form by applications of the axioms in ACG and replacing recursion variables by the right-hand sides of their recursive equations;
  \item game equivalence $\sim$ is a congruences with respect to all game operations of ACG.
\end{enumerate}
\end{proof}

\subsection{Recursive Definition and Specification Principles}

The $RDP$ (Recursive Definition Principle) and the $RSP$ (Recursive Specification Principle) are shown in Table \ref{RDPRSP}.

\begin{center}
\begin{table}
  \begin{tabular}{@{}ll@{}}
\hline No. &Axiom\\
  $RDP$ & $\langle X_i|E\rangle = G_i(\langle X_1|E,\cdots,X_n|E\rangle)\quad (i\in\{1,\cdots,n\})$\\
  $RSP$ & if $y_i=G_i(y_1,\cdots,y_n)$ for $i\in\{1,\cdots,n\}$, then $y_i=\langle X_i|E\rangle \quad(i\in\{1,\cdots,n\})$\\
\end{tabular}
\caption{Recursive definition and specification principle}
\label{RDPRSP}
\end{table}
\end{center}

\begin{theorem}[Elimination theorem of ACG with linear recursion]\label{ETRecursion}
Each game term in ACG with linear recursion is equal to a game term $\langle X_1|E\rangle$ with $E$ a linear recursive specification.
\end{theorem}

\begin{proof}
By applying structural induction with respect to term size, each game term $G_1$ in ACG with linear recursion generates a game can be expressed in the form of equations

$$G_i=(g_{i11}\parallel\cdots\parallel g_{i1i_1})G_{i1}\vee^i\cdots\vee^i(g_{ik_i1}\parallel\cdots\parallel g_{ik_ii_k})G_{ik_i}\vee^i(h_{i11}\parallel\cdots\parallel h_{i1i_1})\vee^i\cdots\vee^i(h_{il_i1}\parallel\cdots\parallel h_{il_ii_l})$$

for $i\in\{1,\cdots,n\}$. Let the linear recursive specification $E$ consist of the recursive equations

$$X_i=(g_{i11}\parallel\cdots\parallel g_{i1i_1})X_{i1}\vee^i\cdots\vee^i(g_{ik_i1}\parallel\cdots\parallel g_{ik_ii_k})X_{ik_i}\vee^i(h_{i11}\parallel\cdots\parallel h_{i1i_1})\vee^i\cdots\vee^i(h_{il_i1}\parallel\cdots\parallel h_{il_ii_l})$$

for $i\in\{1,\cdots,n\}$. Replacing $X_i$ by $G_i$ for $i\in\{1,\cdots,n\}$ is a solution for $E$, $RSP$ yields $G_1=\langle X_1|E\rangle$.
\end{proof}

\begin{theorem}[Soundness of ACG with guarded recursion]\label{SACGR}
Let $x$ and $y$ be ACG with guarded recursion terms. If $ACG\textrm{ with guarded recursion}\vdash x=y$, then $x\sim y$;
\end{theorem}

\begin{proof}
Since game equivalence $\sim$ is both an equivalent and a congruent relation with respect to ACG with guarded recursion, we only need to check if each axiom in Table \ref{RDPRSP} is sound modulo game equivalence. We leave this proof as an exercise for the readers.
\end{proof}

\begin{theorem}[Completeness of ACG with linear recursion]\label{CACGR}
Let $G$ and $H$ be closed ACG with linear recursion terms, then if $G\sim H$ then $G=H$.
\end{theorem}

\begin{proof}
Firstly, by the elimination theorem of ACG with guarded recursion (see Theorem \ref{ETRecursion}), we know that each game term in ACG with linear recursion is equal to a game term $\langle X_1|E\rangle$ with $E$ a linear recursive specification.

It remains to prove that if $\langle X_1|E_1\rangle \sim \langle Y_1|E_2\rangle$ for linear recursive specification $E_1$ and $E_2$, then $\langle X_1|E_1\rangle = \langle Y_1|E_2\rangle$.

Let $E_1$ consist of recursive equations $X=G_X$ for $X\in \mathcal{X}$ and $E_2$
consists of recursion equations $Y=G_Y$ for $Y\in\mathcal{Y}$. Let the linear recursive specification $E$ consist of recursion equations $Z_{XY}=G_{XY}$, and $\langle X|E_1\rangle\sim\langle Y|E_2\rangle$, and $G_{XY}$ consists of the following summands:

\begin{enumerate}
  \item $G_{XY}$ contains a summand $(g_1\parallel\cdots\parallel g_m)Z_{X'Y'}$ iff $G_X$ contains the summand $(g_1\parallel\cdots\parallel g_m)X'$ and $G_Y$ contains the summand $(g_1\parallel\cdots\parallel g_m)Y'$ such that $\langle X'|E_1\rangle\sim\langle Y'|E_2\rangle$;
  \item $G_{XY}$ contains a summand $h_1\parallel\cdots\parallel h_n$ iff $G_X$ contains the summand $h_1\parallel\cdots\parallel h_n$ and $G_Y$ contains the summand $h_1\parallel\cdots\parallel h_n$.
\end{enumerate}

Let $\sigma$ map recursion variable $X$ in $E_1$ to $\langle X|E_1\rangle$, and let $\psi$ map recursion variable $Z_{XY}$ in $E$ to $\langle X|E_1\rangle$. So, $\sigma((g_1\parallel\cdots\parallel g_m)X')\equiv(g_1\parallel\cdots\parallel g_m)\langle X'|E_1\rangle\equiv\psi((g_1\parallel\cdots\parallel g_m)Z_{X'Y'})$, so by $RDP$, we get $\langle X|E_1\rangle=\sigma(G_X)=\psi(G_{XY})$. Then by $RSP$, $\langle X|E_1\rangle=\langle Z_{XY}|E\rangle$, particularly, $\langle X_1|E_1\rangle=\langle Z_{X_1Y_1}|E\rangle$. Similarly, we can obtain $\langle Y_1|E_2\rangle=\langle Z_{X_1Y_1}|E\rangle$. Finally, $\langle X_1|E_1\rangle=\langle Z_{X_1Y_1}|E\rangle=\langle Y_1|E_2\rangle$, as desired.
\end{proof}

\section{Abstraction}\label{abs}

In this section, we consider abstraction to abstract away inner games by use of idle game $\iota$.

\subsection{Guarded Linear Recursion}

The idle game $\iota$ as an atomic game, is introduced into $E$. Considering the recursive specification $X=\iota X$, $\iota G$, $\iota\iota G$, and $\iota\cdots G$ are all its solutions, that is, the solutions make the existence of $\iota$-loops which cause unfairness. To prevent $\iota$-loops, we extend the definition of linear recursive specification (Definition \ref{LRS}) to guarded one.

\begin{definition}[Guarded linear recursive specification]\label{GLRS}
A recursive specification is linear if its recursive equations are of the form

$$(g_{11}\parallel\cdots\parallel g_{1i_1})X_1\vee^i\cdots\vee^i(g_{k1}\parallel\cdots\parallel g_{ki_k})X_k\vee^i(h_{11}\parallel\cdots\parallel h_{1j_1})\vee^i\cdots\vee^i(h_{1j_1}\parallel\cdots\parallel h_{lj_l})$$

where $g_{11},\cdots,g_{1i_1},g_{k1},\cdots,g_{ki_k},h_{11},\cdots,h_{1j_1},h_{1j_1},\cdots,h_{lj_l}\in \mathcal{G}_{at}\cup\{\iota\}$.

A linear recursive specification $E$ is guarded if there does not exist an infinite sequence of $\iota$-transitions $\langle X|E\rangle\xrightarrow{\iota}\langle X'|E\rangle\xrightarrow{\iota}\langle X''|E\rangle\xrightarrow{\iota}\cdots$.
\end{definition}

\begin{theorem}[Conservitivity of ACG with idle game and guarded linear recursion]
ACG with idle game and guarded linear recursion is a conservative extension of ACG with linear recursion.
\end{theorem}

\begin{proof}
Since the transition rules of ACG with linear recursion are source-dependent, and the transition rules for idle game in Table \ref{TRForBAG} contain only a fresh constant $\iota$ in their source, so the transition rules of ACG with idle game and guarded linear recursion is a conservative extension of those of ACG with linear recursion.
\end{proof}

\begin{theorem}[Congruence theorem of ACG with idle game and guarded linear recursion]
weak game equivalence $\approx$ is a congruence with respect to ACG with idle game and guarded linear recursion.
\end{theorem}

\begin{proof}
It follows the following two facts:
\begin{enumerate}
  \item in a guarded linear recursive specification, right-hand sides of its recursive equations can be adapted to the form by applications of the axioms in ACG and replacing recursion variables by the right-hand sides of their recursive equations;
  \item weak game equivalence $\approx$ is a congruence with respect to all game operations of ACG.
\end{enumerate}
\end{proof}

The axioms and transition rules of the idle game $\iota$, please see Section \ref{bag} and Section \ref{acg}.

\subsection{Abstraction}

The unary abstraction operator $\iota_I$ ($I\subseteq \mathcal{G}_{at}$) renames all atomic games in $I$ into $\iota$. ACG with idle game and abstraction operator is called $ACG_{\iota}$. The transition rules of operator $\iota_I$ are shown in Table \ref{TRForAbstraction}.

\begin{center}
    \begin{table}
        $$\frac{x\xrightarrow{g_a}\surd}{\iota_I(x)\xrightarrow{g_a}\surd}\quad g_a\notin I
        \quad\quad\frac{x\xrightarrow{g_a}x'}{\iota_I(x)\xrightarrow{g_a}\iota_I(x')}\quad g_a\notin I$$

        $$\frac{x\xrightarrow{g_a}\surd}{\iota_I(x)\rightarrow\surd}\quad g_a\in I
        \quad\quad\frac{x\xrightarrow{g_a}x'}{\iota_I(x)\rightarrow\iota_I(x')}\quad g_a\in I$$
        \caption{Transition rule of the abstraction operator}
        \label{TRForAbstraction}
    \end{table}
\end{center}

\begin{theorem}[Conservitivity of $ACG_{\iota}$ with guarded linear recursion]
$ACG_{\iota}$ with guarded linear recursion is a conservative extension of ACG with idle game and guarded linear recursion.
\end{theorem}

\begin{proof}
Since the transition rules of ACG with idle game and guarded linear recursion are source-dependent, and the transition rules for abstraction operator in Table \ref{TRForAbstraction} contain only a fresh operator $\iota_I$ in their source, so the transition rules of $ACG_{\iota}$ with guarded linear recursion is a conservative extension of those of ACG with idle game and guarded linear recursion.
\end{proof}

\begin{theorem}[Congruence theorem of $ACG_{\iota}$ with guarded linear recursion]
Weak game equivalence $\approx$ is a congruence with respect to $ACG_{\iota}$ with guarded linear recursion.
\end{theorem}

\begin{proof}
Suppose that $G_1\approx G_2$, it is sufficient to prove that $\iota_I(G_1)\approx\iota_I(G_2)$. The proof can be immediately gotten from the definition of weak game equivalence (Definition \ref{dwge}) and the transition rules in Table \ref{TRForAbstraction}.
\end{proof}

We design the axioms for the abstraction operator $\iota_I$ in Table \ref{AxiomsForAbstraction}.

\begin{center}
\begin{table}
  \begin{tabular}{@{}ll@{}}
\hline No. &Axiom\\
  $II1$ & $g_a\notin I\quad \iota_I(g_a)=g_a$\\
  $II2$ & $g_a\in I\quad \iota_I(g_a)=\iota$\\
  $II3$ & $\iota_I(x\wedge y)=\iota_I(x)\wedge\iota_I(y)$\\
  $II4$ & $\iota_I(x\vee y)=\iota_I(x)\vee\iota_I(y)$\\
  $II5$ & $\iota_I(x\circ y)=\iota_I(x)\circ\iota_I(y)$\\
  $II6$ & $\iota_I(x\parallel y)=\iota_I(x)\parallel\iota_I(y)$\\
  $II7$ & $\iota_I(x^d)=(\iota_I(x))^d$\\
\end{tabular}
\caption{Axioms of abstraction operator}
\label{AxiomsForAbstraction}
\end{table}
\end{center}

\begin{theorem}[Soundness of $ACG_{\iota}$ with guarded linear recursion]\label{SACGABS}
Let $x$ and $y$ be $ACG_{\iota}$ with guarded linear recursion terms. If $ACG_{\iota}$ with guarded linear recursion $\vdash x=y$, then $x\approx y$.
\end{theorem}

\begin{proof}
Since $\approx$ is both an equivalent and a congruent relation with respect to $ACG_{\iota}$ with guarded linear recursion, we only need to check if each axiom in Table \ref{AxiomsForAbstraction} is sound modulo $\approx$. The proof is left to the readers as an exercise.
\end{proof}

Though $\iota$-loops are prohibited in guarded linear recursive specifications (see Definition \ref{GLRS}) in a specifiable way, they can be constructed using the abstraction operator, for example, there exist $\iota$-loops in the game term $\iota_{\{g_a\}}(\langle X|X=g_aX\rangle)$. To avoid $\iota$-loops caused by $\iota_I$ and ensure fairness, the concept of cluster and $CFAR$ (Cluster Fair Abstraction Rule) \cite{CFAR} are still valid in games, we introduce them below.

\begin{definition}[Cluster]\label{CLUSTER}
Let $E$ be a guarded linear recursive specification, and $I\subseteq \mathcal{G}_{at}$. Two recursion variable $X$ and $Y$ in $E$ are in the same cluster for $I$ iff there exist sequences of transitions $\langle X|E\rangle\xrightarrow{\{h_{11},\cdots, h_{1i}\}}\cdots\xrightarrow{\{h_{m1},\cdots, h_{mi}\}}\langle Y|E\rangle$ and $\langle Y|E\rangle\xrightarrow{\{c_{11},\cdots, c_{1j}\}}\cdots\xrightarrow{\{c_{n1},\cdots, c_{nj}\}}\langle X|E\rangle$, where $h_{11},\cdots,h_{mi},c_{11},\cdots,c_{nj}\in I\cup\{\iota\}$.

$g_1\parallel\cdots\parallel g_k$ or $(g_1\parallel\cdots\parallel g_k) X$ is an exit for the cluster $C$ iff: (1) $g_1\parallel\cdots\parallel g_k$ or $(g_1\parallel\cdots\parallel g_k) X$ is a summand at the right-hand side of the recursive equation for a recursion variable in $C$, and (2) in the case of $(g_1\parallel\cdots\parallel g_k) X$, either $g_l\notin I\cup\{\iota\}(l\in\{1,2,\cdots,k\})$ or $X\notin C$.
\end{definition}

\begin{center}
\begin{table}
  \begin{tabular}{@{}ll@{}}
  \hline
  No. &Axiom\\
  $CFAR$ & If $X$ is in a cluster for $I$ with exits \\
           & $\{(g_{11}\parallel\cdots\parallel g_{1i})Y_1,\cdots,(g_{m1}\parallel\cdots\parallel g_{mi})Y_m, h_{11}\parallel\cdots\parallel h_{1j},\cdots,h_{n1}\parallel\cdots\parallel h_{nj}\}$, \\
           & then $\iota\circ\iota_I(\langle X|E\rangle)=$\\
           & $\iota\circ\iota_I((g_{11}\parallel\cdots\parallel g_{1i})\langle Y_1|E\rangle\vee^i\cdots\vee^i(g_{m1}\parallel\cdots\parallel g_{mi})\langle Y_m|E\rangle\vee^i h_{11}\parallel\cdots\parallel h_{1j}\vee^i\cdots\vee^i h_{n1}\parallel\cdots\parallel h_{nj})$\\
\end{tabular}
\caption{Cluster fair abstraction rule}
\label{CFAR}
\end{table}
\end{center}

\begin{theorem}[Soundness of $CFAR$]\label{SCFAR}
$CFAR$ is sound modulo weak game equivalence $\approx$.
\end{theorem}

\begin{proof}
Let $X$ be in a cluster for $I$ with exits $\{(g_{11}\parallel\cdots\parallel g_{1i})Y_1,\cdots,(g_{m1}\parallel\cdots\parallel g_{mi})Y_m,h_{11}\parallel\cdots\parallel h_{1j},\cdots,h_{n1}\parallel\cdots\parallel h_{nj}\}$. Then $\langle X|E\rangle$ can play a string of atomic games from $I\cup\{\iota\}$ inside the cluster of $X$, followed by an exit $(g_{i'1}\parallel\cdots\parallel g_{i'i})Y_{i'}$ for $i'\in\{1,\cdots,m\}$ or $h_{j'1}\parallel\cdots\parallel h_{j'j}$ for $j'\in\{1,\cdots,n\}$. Hence, $\iota_I(\langle X|E\rangle)$ can play a string of $\iota^*$ inside the cluster of $X$, followed by an exit $\iota_I((g_{i'1}\parallel\cdots\parallel g_{i'i})\langle Y_{i'}|E\rangle)$ for $i'\in\{1,\cdots,m\}$ or $\iota_I(h_{j'1}\parallel\cdots\parallel h_{j'j})$ for $j'\in\{1,\cdots,n\}$. And these $\iota^*$ are non-initial in $\iota\iota_I(\langle X|E\rangle)$, so they are truly idle, we obtain $\iota\iota_I(\langle X|E\rangle)\approx\iota\circ\iota_I((g_{11}\parallel\cdots\parallel g_{1i})\langle Y_1|E\rangle\vee^i\cdots\vee^i(g_{m1}\parallel\cdots\parallel g_{mi})\langle Y_m|E\rangle\vee^i h_{11}\parallel\cdots\parallel h_{1j}\vee^i\cdots\vee^i b_{n1}\parallel\cdots\parallel h_{nj})$, as desired.
\end{proof}

\begin{theorem}[Completeness of $ACG_{\iota}$ with guarded linear recursion and $CFAR$]\label{CCFAR}
Let $G$ and $H$ be closed $ACG_{\iota}$ with guarded linear recursion and $CFAR$ terms, then, if $G\approx H$ then $G=H$.
\end{theorem}

\begin{proof}
Firstly, we know that each process term $G$ in ACG with idle game and guarded linear recursion is equal to a process term $\langle X_1|E\rangle$ with $E$ a guarded linear recursive specification. And we prove if $\langle X_1|E_1\rangle\approx\langle Y_1|E_2\rangle$, then $\langle X_1|E_1\rangle=\langle Y_1|E_2\rangle$

The only new case is $G\equiv\iota_I(H)$. Let $H=\langle X|E\rangle$ with $E$ a guarded linear recursive specification, so $G=\iota_I(\langle X|E\rangle)$. Then the collection of recursive variables in $E$ can be divided into its clusters $C_1,\cdots,C_N$ for $I$. Let

$$(g_{1i1}\parallel\cdots\parallel g_{k_{i1}i1}) Y_{i1}\vee^i\cdots\vee^i(g_{1im_i}\parallel\cdots\parallel g_{k_{im_i}im_i}) Y_{im_i}\vee^i h_{1i1}\parallel\cdots\parallel h_{l_{i1}i1}\vee^i\cdots\vee^i h_{1im_i}\parallel\cdots\parallel h_{l_{im_i}im_i}$$

be the conflict composition of exits for the cluster $C_i$, with $i\in\{1,\cdots,N\}$.

For $Z\in C_i$ with $i\in\{1,\cdots,N\}$, we define

$$G_Z\triangleq (\hat{g_{1i1}}\parallel\cdots\parallel \hat{g_{k_{i1}i1}}) \iota_I(\langle Y_{i1}|E\rangle)\vee^i\cdots\vee^i(\hat{g_{1im_i}}\parallel\cdots\parallel \hat{g_{k_{im_i}im_i}}) \iota_I(\langle Y_{im_i}|E\rangle)\vee^i\hat{h_{1i1}}\parallel\cdots\parallel \hat{h_{l_{i1}i1}}\vee^i\cdots\vee^i\hat{h_{1im_i}}\parallel\cdots\parallel \hat{h_{l_{im_i}im_i}}$$

For $Z\in C_i$ and $g_1,\cdots,g_j\in \mathcal{G}_{at}\cup\{\iota\}$ with $j\in\mathbb{N}$, we have

$(g_1\parallel\cdots\parallel g_j)\iota_I(\langle Z|E\rangle)$

$=(g_1\parallel\cdots\parallel g_j)\iota_I((g_{1i1}\parallel\cdots\parallel g_{k_{i1}i1}) \langle Y_{i1}|E\rangle\vee^i\cdots\vee^i(g_{1im_i}\parallel\cdots\parallel g_{k_{im_i}im_i}) \langle Y_{im_i}|E\rangle\vee^i h_{1i1}\parallel\cdots\parallel h_{l_{i1}i1}\vee^i\cdots\vee^i h_{1im_i}\parallel\cdots\parallel h_{l_{im_i}im_i})$

$=(g_1\parallel\cdots\parallel g_j)s_Z$

Let the linear recursive specification $F$ contain the same recursive variables as $E$, for $Z\in C_i$, $F$ contains the following recursive equation

$$Z=(\hat{g_{1i1}}\parallel\cdots\parallel \hat{g_{k_{i1}i1}}) Y_{i1}\vee^i\cdots\vee^i(\hat{g_{1im_i}}\parallel\cdots\parallel \hat{g_{k_{im_i}im_i}})  Y_{im_i}\vee^i\hat{h_{1i1}}\parallel\cdots\parallel \hat{h_{l_{i1}i1}}\vee^i\cdots\vee^i\hat{h_{1im_i}}\parallel\cdots\parallel \hat{h_{l_{im_i}im_i}}$$

It is easy to see that there is no sequence of one or more $\iota$-transitions from $\langle Z|F\rangle$ to itself, so $F$ is guarded.

For

$$G_Z=(\hat{g_{1i1}}\parallel\cdots\parallel \hat{g_{k_{i1}i1}}) Y_{i1}\vee^i\cdots\vee^i(\hat{g_{1im_i}}\parallel\cdots\parallel \hat{g_{k_{im_i}im_i}}) Y_{im_i}\vee^i\hat{h_{1i1}}\parallel\cdots\parallel \hat{h_{l_{i1}i1}}\vee^i\cdots\vee^i\hat{h_{1im_i}}\parallel\cdots\parallel \hat{h_{l_{im_i}im_i}}$$

is a solution for $F$. So, $(g_1\parallel\cdots\parallel g_j)\iota_I(\langle Z|E\rangle)=(g_1\parallel\cdots\parallel g_j)s_Z=(g_1\parallel\cdots\parallel g_j)\langle Z|F\rangle$.

So,

$$\langle Z|F\rangle=(\hat{g_{1i1}}\parallel\cdots\parallel \hat{g_{k_{i1}i1}}) \langle Y_{i1}|F\rangle\vee^i\cdots\vee^i(\hat{g_{1im_i}}\parallel\cdots\parallel \hat{g_{k_{im_i}im_i}}) \langle Y_{im_i}|F\rangle\vee^i\hat{h_{1i1}}\parallel\cdots\parallel \hat{h_{l_{i1}i1}}\vee^i\cdots\vee^i\hat{h_{1im_i}}\parallel\cdots\parallel \hat{h_{l_{im_i}im_i}}$$

Hence, $\iota_I(\langle X|E\rangle=\langle Z|F\rangle)$, as desired.
\end{proof}

\section{Conclusions}\label{con}

We introduce operational semantics into games in this paper. And based on the operational semantics, we extend the basic algebra of games \cite{BAG1} \cite{BAG2} and algebra of concurrent games \cite{CG4} with recursion and abstraction, and establish a fully algebraic axiomatization for games.

The future work will include two aspects: one is to beyond two-person game to establish algebraic theories of more kind of games, such as imperfect games; the other is that the algebras can be used to reason on the behaviors of systems (not only computational systems) with game theory supported.

\newpage

%\appendix
%\section{appendix 1}
%
%appendix 1

\label{lastpage}

\end{document}